\documentclass[runningheads]{llncs}
\usepackage{todonotes}
\usepackage{amssymb}
\usepackage{enumitem}
\usepackage{graphicx}
\usepackage{hyperref}

\usepackage{algorithm}
\usepackage{algpseudocode}

\algrenewcommand\algorithmicrequire{\textbf{Input :}}
\algrenewcommand\algorithmicensure{\textbf{Output :}}

\algnewcommand{\lIf}[2]{
  \State \algorithmicif\ #1\ \algorithmicthen\ #2}

\algnewcommand{\lElsIf}[2]{
  \State \algorithmicelse\ \algorithmicif\ #1\ \algorithmicthen\ #2}

\algnewcommand{\lElse}[1]{
  \State \algorithmicelse\ #1}
  
\algnewcommand{\Ret}[1]{
  \State \algorithmicreturn\ #1}

\algnewcommand{\IfThenElse}[3]{
  \State \algorithmicif\ #1\ \algorithmicthen\ #2\ \algorithmicelse\ #3}

\usepackage{mathtools}

\newcommand{\bO}{{\cal O}}
\newcommand{\tO}{\tilde{\bO}}
\newcommand{\sO}{\bO^{\ast}}

\DeclarePairedDelimiter{\braces}{\{}{\}}		    
\DeclarePairedDelimiter{\bracks}{[}{]}		        
\DeclarePairedDelimiter{\parens}{(}{)}		        
\DeclarePairedDelimiter{\abs}{\lvert}{\rvert}		
\DeclarePairedDelimiterX{\setdef}[2]{\{}{\}}{#1 \mid #2}		
\DeclarePairedDelimiterXPP{\exclude}[1]{\mathopen{}\setminus}{\{}{\}}{}{#1}
\newcommand{\Part}{\textsc{Partition}}
\newcommand{\SubS}{\textsc{Subset Sum}}
\newcommand{\ESS}{\textsc{Equal Subset Sum}}

\newcommand{\SSR}{\textsc{Subset Sum Ratio}}

\newcommand{\etal}{\textsl{et al.}}

\newcommand{\opt}{\text{\normalfont opt}}
\newcommand{\myorcidID}[1]{\href{https://orcid.org/#1}{\includegraphics[height=9pt]{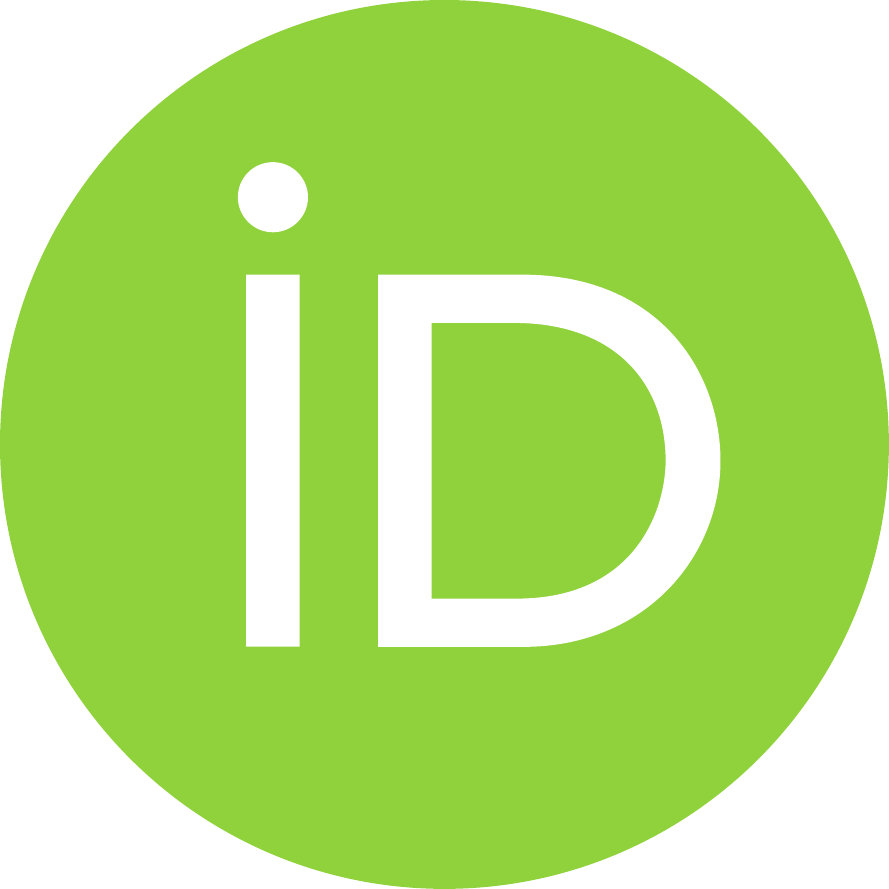}}}
\begin{document}
\title{Approximating Subset Sum Ratio via Partition Computations\thanks{Aris Pagourtzis and Stavros Petsalakis were supported in part by the PEVE 2020 basic research support program of the National Technical University of Athens.}}
%
%
\author{
    Giannis Alonistiotis\inst{1}\myorcidID{0000-0003-1277-5017} \and
    Antonis Antonopoulos\inst{1}\myorcidID{0000-0002-1368-6334} \and
    Nikolaos Melissinos\inst{2}\myorcidID{0000-0002-0864-9803}  \and
    Aris Pagourtzis\inst{1}\myorcidID{0000-0002-6220-3722}      \and
    Stavros Petsalakis\inst{1}\myorcidID{0000-0001-7825-2839}   \and
    Manolis Vasilakis\inst{2}\myorcidID{0000-0001-6505-2977}
}
\authorrunning{G. Alonistiotis et al.}
\institute{
    School of Electrical and Computer Engineering, National Technical University of Athens, Polytechnioupoli, 15780 Zografou, Athens, Greece\\
    \email{ \{ialonistiotis,aanton,spetsalakis\}@corelab.ntua.gr, pagour@cs.ntua.gr} \and
    Universit\'{e} Paris-Dauphine, PSL University, CNRS, LAMSADE, 75016 Paris, France\\
    \email{ \{nikolaos.melissinos,emmanouil.vasilakis\}@dauphine.eu}
}
\maketitle              
\begin{abstract}
    We present a new FPTAS for the \textsc{Subset Sum Ratio} problem,
    which, given a set of integers, asks for two disjoint subsets such that the ratio of their sums is as close to $1$ as possible.
    Our scheme makes use of exact and approximate algorithms for the closely related \textsc{Partition} problem,
    hence any progress over those---such as the recent improvement due to Bringmann and Nakos [SODA 2021]---carries over to our FPTAS.
    Depending on the relationship between the size of the input set $n$ and the error margin $\varepsilon$,
    we improve upon the best currently known algorithm of Melissinos and Pagourtzis [COCOON 2018] of complexity $\mathcal{O} (n^4 / \varepsilon)$.
    In particular, the exponent of $n$ in our proposed scheme may decrease down to $2$, depending on the \textsc{Partition} algorithm used.
    Furthermore, while the aforementioned state of the art complexity expressed in the form $\mathcal{O} ((n + 1 / \varepsilon)^c)$
    has constant $c = 5$, our results establish that $c < 5$.
    
\end{abstract}

\section{Introduction} \label{sec:intro}

One of Karp's $21$ NP-complete problems~\cite{Karp72}, \SubS\ has seen astounding 
progress over the last few years.
Koiliaris and Xu~\cite{KoiliarisX19}, Bringmann~\cite{Bringmann17} and Jin and Wu~\cite{JinW19} have presented pseudopolynomial algorithms resulting in substantial improvements over the long-standing standard approach of Bellman~\cite{Bellman}, and the improvement by Pisinger~\cite{Pisinger99}.
Moreover, the latter two algorithms~\cite{Bringmann17,JinW19} match the SETH-based lower bounds proved in~\cite{AbboudBHS22}.
Additionally, recently there has been progress in the approximation scheme of \SubS, the first such improvement in over 20 years, with a new algorithm introduced by Bringmann and Nakos~\cite{BringmannN21}, as well as corresponding lower bounds obtained through the lens of fine-grained complexity.

A thoroughly studied special case of \SubS\ is the \Part\ problem,
which asks for a partition of the input set to two subsets such that the difference of their sums is minimum.
Any algorithm solving the first applies to the latter, though recent progress~\cite{BringmannN21,MuchaW019} has shown that \Part\ may be solved more efficiently in the approximation setting.
On the other hand, regarding exact solutions, no better algorithm has been developed, therefore \SubS\ algorithms remain the state of the art.

The \ESS\ problem, which, given an input set, asks for two disjoint subsets of equal sum, is closely related to \SubS\ and \Part.
It finds applications in multiple different fields, ranging from computational biology~\cite{CieliebakEP03,CieliebakE04} and computational social choice~\cite{LiptonMMS04}, to cryptography~\cite{Vol17}, to name a few.
In addition, it is related to important theoretical concepts such as the complexity of search problems in the class \textsf{TFNP}~\cite{Papadimitriou94}.

The centerpiece of this paper is the \SSR\ problem, the optimization version of \ESS, which asks, given an input set $S \subseteq \mathbb{N}$, for two disjoint subsets $S_1, S_2 \subseteq S$, such that the following ratio is minimized
\[
    \frac{\max \braces*{\sum_{s_i \in S_1} s_i, \sum_{s_j \in S_2} s_j}}
    {\min \braces*{\sum_{s_i \in S_1} s_i, \sum_{s_j \in S_2} s_j}}
\]

We present a new approximation scheme for \SSR, highlighting its close relationship with the classical \Part\ problem.
Our proposed algorithm is the first to associate these closely related problems and, depending on the relationship of the cardinality of the input set $n$ and the value of the error margin $\varepsilon$, achieves better asymptotic bounds than the current state of the art~\cite{MelissinosP18}.
Moreover, while the complexity of the current state of the art approximation scheme expressed in the form $\bO \parens{\parens{n + 1 / \varepsilon}^c}$ has an exponent $c = 5$,
we present an FPTAS with constant $c < 5$.

\subsection{Related Work}

\ESS\ as well as its optimization version called \SSR~\cite{BazganST02} are closely related to problems appearing in many scientific areas.
Some examples include the \textsc{Partial Digest} problem, which comes from computational biology~\cite{CieliebakEP03,CieliebakE04}, the
allocation of individual goods~\cite{LiptonMMS04}, tournament construction~\cite{khanphd}, and a variation of \SubS, called Multiple Integrated Sets SSP, which finds applications in the field of cryptography~\cite{Vol17}.
Furthermore, it is related to important concepts in theoretical computer science; for example, a restricted version of \ESS\ lies in a subclass of the complexity class $\mathsf{TFNP}$, namely in $\mathsf{PPP}$~\cite{Papadimitriou94}, a class consisting of search problems that always have a solution due to some pigeonhole argument, and no polynomial time algorithm is known for this restricted version.

\ESS\ has been proven NP-hard by Woeginger and Yu~\cite{WoegingerY92} (see also the full version of~\cite{MuchaNPW19} for an alternative proof) and several variations have been proven NP-hard by Cieliebak \textsl{et al.} in~\cite{CieliebakEP03b,CieliebakEPS08}.
A 1.324-approximation algorithm has been proposed for \SSR\ in~\cite{WoegingerY92} 
and several FPTASs appeared in~\cite{BazganST02,Nanongkai13,MelissinosP18}, the fastest so far being the one in~\cite{MelissinosP18} of complexity $\bO (n^4/\varepsilon)$, the complexity of which seems to also apply to various meaningful special cases, as shown in~\cite{MelissinosPT22}.

As far as exact algorithms are concerned, recent progress has shown that \ESS\ can be solved probabilistically
in\footnote{Standard $\sO$ and $\tO$ notation is used to hide polynomial and polylogarithmic factors respectively.} $\sO (1.7088^n)$ time~\cite{MuchaNPW19},
faster than a standard ``meet-in-the-middle'' approach yielding an $\sO (3^{n/2}) \le \sO (1.7321^n)$ time algorithm.

These problems are tightly connected to \SubS, which has seen impressive advances recently, due to Koiliaris and Xu~\cite{KoiliarisX19} who gave a deterministic $\tO (\sqrt{n}t)$ algorithm, where $n$ is the number of input elements and $t$ is the target, and Bringmann~\cite{Bringmann17} who gave a $\tO(n + t)$ randomized algorithm,
which is essentially
optimal under SETH~\cite{AbboudBHS22}.
See also~\cite{AntonopoulosPPV22} for an extension of these algorithms to a more general setting.
Jin and Wu subsequently proposed a simpler randomized algorithm~\cite{JinW19} achieving the same bounds as~\cite{Bringmann17}, which however seems to only solve the decision version of the problem.
Recently, Bringmann and Nakos~\cite{BringmannN20} have presented an $\bO \parens*{\lvert \mathcal{S}_t(Z) \rvert^{4/3} {\rm poly}(\log t)}$ algorithm, where $\mathcal{S}_t(Z)$ is the set of all subset sums of the input set $Z$ that are smaller than $t$, based on top-$k$ convolution.

\Part\ shares the complexity of \SubS\ regarding exact solutions,
where the meet in the middle approach~\cite{HorowitzS74} from the 70's remains the state of the art as far as algorithms dependent on $n$ are concerned.
On the other hand, one can approximate \Part\ more efficiently than \SubS\ unless the min-plus convolution conjecture~\cite{CyganMWW19} is false.
In particular, Bringmann and Nakos~\cite{BringmannN21} have presented the first improvement for the latter in over 20 years,
since the scheme of~\cite{KellererMPS03} had remained the state of the art.
Moreover, in their paper they have shown that developing a significantly better algorithm would contradict said conjecture.
Furthermore, they develop an approximation scheme for \Part\ utilizing min-plus convolution computations, improving upon the recent work of Mucha \etal~\cite{MuchaW019}
and circumventing the lower bounds established for \SubS\ in their work.

\subsection{Our Contribution}

We present a novel approximation scheme for the \SSR\ problem.
Our algorithm makes use of exact and approximation algorithms for \Part, thus, any improvement over those carries over to our proposed scheme.
Additionally, depending on the relationship between $n$ and $\varepsilon$, our algorithm improves upon the best existing approximation scheme of~\cite{MelissinosP18}.

We start by presenting some necessary background in Section~\ref{sec:prelim}.
Afterwards, in Section~\ref{sec:algo_restricted} we introduce an FPTAS for a restricted version of the problem.
In the following Section~\ref{sec:algo}, we explain how to make use of the algorithm presented in the previous section, in order to obtain an approximation scheme for the \SSR\ problem.
The complexity of the final scheme is thoroughly analyzed in Section~\ref{sec:complexity},
followed by some possible directions for future research in Section~\ref{sec:future}.
\newline\newline
\noindent\textbf{Prior Work.}
In this current paper we improve upon the results of the preliminary version~\cite{AlonistiotisAMP22},
by using approximate and exact \Part\ algorithms instead of \SubS\ computations.


\section{Preliminaries} \label{sec:prelim}
Let, for $x \in \mathbb{N}$, $[x] = \braces{0, \ldots, x}$ denote the set of integers in the interval $[0, x]$.
Given a set $S \subseteq \mathbb{N}$, denote its \emph{largest element} by $\max(S)$ and the sum of its elements by $\Sigma (S) = \sum_{s \in S} s$.
If we are additionally given a value $\varepsilon \in (0, 1)$, define the following \emph{partition} of its elements:
\begin{itemize}
    \item The set of its \emph{large} elements as $L(S, \varepsilon) = \setdef{s \in S}{s \geq \varepsilon \cdot \max(S)}$.
    Note that $\max(S) \in L(S, \varepsilon)$, for any $\varepsilon \in (0, 1)$.
    
    \item The set of its \emph{small} elements as $M(S, \varepsilon) = \setdef{s \in S}{s < \varepsilon \cdot \max(S)}$.
\end{itemize}
\noindent
In the following, since the values of the associated parameters will be clear from the context,
they will be omitted and we will refer to these sets simply as $L$ and $M$.


\begin{definition}[\Part]
Given a set $X$, compute a subset $X^*_p \subseteq X$, such that $\Sigma (X^*_p) = \max \braces*{ \Sigma (Z) \mid Z \subseteq X, \Sigma (Z) \leq \Sigma(X) / 2}$.
Moreover, let $\overline{X^*_p} = X \setminus X^*_p$.
\end{definition}

\begin{definition}[Approximate \Part, from \cite{MuchaW019}]
Given a set $X$ and error margin $\varepsilon$, compute a subset $X_p \subseteq X$ such that $(1 - \varepsilon) \cdot \Sigma (X^*_p) \leq \Sigma (X_p) \leq \Sigma (X^*_p)$.
Moreover, let $\overline{X_p} = X \setminus X_p$.
\end{definition}



\section{Scheme for a Restricted Version} \label{sec:algo_restricted}
In this section, we present an FPTAS for the constrained version of the \SSR\ problem where we are only interested in approximating solutions that involve the largest element of the input set.
In other words, one of the subsets of the optimal solution contains $\max(A) = a_n$ (assuming that $A = \braces{a_1, \ldots, a_n}$ is the \emph{sorted} input set);
let $r_{\opt}$ denote the subset sum ratio of such an optimal solution.
Our FPTAS will return a solution of ratio $r$, such that $1 \le r \le (1 + \varepsilon) \cdot r_{\opt}$, for a given error margin $\varepsilon \in (0, 1)$; however, we allow that the sets of the returned solution do not necessarily satisfy the aforementioned constraint (i.e.\ $a_n$ may not be involved in the approximate solution).

\subsection{Outline of the Algorithm}
We now present a rough outline of the algorithm~\ref{alg:outline}:
\begin{itemize}
    \item At first, we search for approximate solutions involving exclusively large elements from $L(A,\varepsilon)$.

    \item To this end, we produce the subset sums formed by these large elements.
    If their number exceeds $n / \varepsilon^2$, then we can find an approximate solution.
    
    \item Otherwise, there are at most $n / \varepsilon^2$ subsets of large elements.
    In this case, we can find a solution by running an exact or an approximate \Part\ algorithm for each subset.
    
    \item In the case that the optimal solution involves small elements, we show that it suffices to add elements of $M(A, \varepsilon)$ in a greedy way.
\end{itemize}

\begin{algorithm}[H]
	\caption{ConstrainedSSR($A, \varepsilon, T$)}
	\label{alg:outline}
    \begin{algorithmic}[1]
    	\Require Set $A = \braces{a_1, \ldots, a_n}$, error margin $\varepsilon$ and table of partial sums $T$.
    	\Ensure $(1 + \varepsilon)$-apx of the optimal solution respecting the constraint.
    	
    	\State Partition $A$ to $M = \setdef{a_i \in A}{a_i < \varepsilon \cdot a_n}$ and $L = \setdef{a_i \in A}{a_i \geq \varepsilon \cdot a_n}$.

    	\State Split interval $\bracks{0, n \cdot a_n}$ to $n / \varepsilon^2$ bins of size $\varepsilon^2 \cdot a_n$.

    	\While{filling the bins with the subset sums of $L$}
        	\If{two subset sums correspond to the same bin}
        	    \State \Return{an apx solution based on these.} \Comment{$\bO (n / \varepsilon^2)$ complexity.}
            \EndIf
        \EndWhile
	
	    \State $2^{\abs{L}} \leq n / \varepsilon^2 \iff \abs{L} \leq \log (n / \varepsilon^2)$.

        \For{each subset of large elements containing $a_n$} \Comment{$\bO (n / \varepsilon^2)$ subsets.}
            \State Solve corresponding \Part\ instance. \Comment{Complexity in Section~\ref{sec:complexity}.}
            \State Add small elements. \Comment{$\bO (\log n)$ complexity, see Subsection~\ref{subsec:plus_small}.}
        \EndFor
    \end{algorithmic}
\end{algorithm}

\subsection{Solution Involving Exclusively Large Elements}
We firstly search for an $(1 + \varepsilon)$-approximate solution with $\varepsilon \in (0,1)$, without involving any of the elements that are smaller than $\varepsilon \cdot a_n$.
Let $M = \setdef{a_i \in A}{a_i < \varepsilon \cdot a_n}$ be the set of small elements and $L = A \setminus M = \setdef{a_i \in A}{a_i \geq \varepsilon \cdot a_n}$ be the set of large elements.

After partitioning the input set, we split the interval $[0, n \cdot a_n]$ into smaller intervals, called bins, of size $l = \varepsilon^2 \cdot a_n$ each, as depicted in figure~\ref{fig:bins}.

\begin{figure}[htb]
    \centering
    \includegraphics[width=0.8\textwidth]{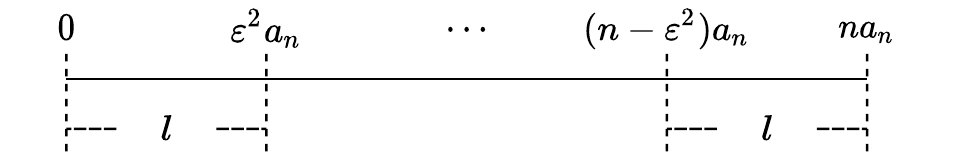}
    \caption{Split of the interval $\bracks{0, n \cdot a_n}$ to bins of size $l$.}
    \label{fig:bins}
\end{figure}

Thus, there are a total of $B = n / \varepsilon^2$ bins.
Notice that each possible subset of the input set will belong to a respective bin constructed this way, depending on its sum.
Additionally, if two sets correspond to the same bin, then the difference of their subset sums will be at most $l$.

The next step of our algorithm is to generate all the possible subset sums, occurring from the set of large elements $L$.
The complexity of this procedure is $\bO \parens*{2^{\abs{L}}}$, where $\abs{L}$ is the cardinality of set $L$.
Notice however, that it is possible to bound the number of the produced subset sums by the number of bins $B$, since if two sums belong to the same bin they constitute a solution, as shown in Lemma~\ref{lem:two_sets_one_bin}, in which case the algorithm terminates in time $\bO (n / \varepsilon^2)$. 

\begin{lemma} \label{lem:two_sets_one_bin}
    If two subsets correspond to the same bin, we can find an $(1 + \varepsilon)$-approximation solution.
\end{lemma}

\begin{proof}
    Suppose there exist two sets $L_1, L_2 \subseteq L$ whose sums correspond to the same bin, with $\Sigma(L_1) \leq \Sigma(L_2)$.
    Notice that there is no guarantee regarding the disjointness of said subsets, thus consider $L'_1 = L_1 \setminus L_2$ and $L'_2 = L_2 \setminus L_1$, for which it is obvious that $\Sigma(L_1') \leq \Sigma(L_2')$.
    Additionally, assume that $L'_1 \neq \emptyset$.
    Then it holds that
    \[
        \Sigma(L_2') - \Sigma(L_1') =
        \Sigma(L_2) - \Sigma(L_1) \leq l.
    \]
    Therefore, the sets $L'_1$ and $L'_2$ constitute an $(1+\varepsilon)$-approximation solution, since
    \begin{align*}
        \frac{\Sigma(L_2')}{\Sigma(L_1')} &\leq
        \frac{\Sigma(L_1') + l}{\Sigma(L_1')}
        = 1 + \frac{l}{\Sigma(L_1')}\\
        &\leq 1 + \frac{\varepsilon^2 \cdot a_n}{\varepsilon \cdot a_n}
        = 1 + \varepsilon
    \end{align*}
    where the last inequality is due to the fact that $L_1' \subseteq L$ is composed of elements $\geq \varepsilon \cdot a_n$, thus $\Sigma(L_1') \geq \varepsilon \cdot a_n$.

    It remains to show that $L'_1 \neq \emptyset$.
    Assume that $L'_1 = \emptyset$.
    This implies that $L_1 \subseteq L_2$ and since we consider each subset of $L$ only once and the input is a set and not a multiset, it holds that $L_1 \subset L_2 \implies L'_2 \neq \emptyset$.
    Since $L_1$ and $L_2$ correspond to the same bin, it holds that
    \[
        \Sigma(L_2) - \Sigma(L_1) \leq l \implies
        \Sigma(L_2') - \Sigma(L_1') \leq l \implies
        \Sigma(L'_2) \leq l
    \]
    which is a contradiction, since $L'_2$ is a non empty subset of $L$, which is comprised of elements greater than or equal to $\varepsilon \cdot a_n$, hence $\Sigma (L'_2) \geq \varepsilon \cdot a_n > \varepsilon^2 \cdot a_n = l$, since $\varepsilon < 1$.
\end{proof}

Consider an $\varepsilon'$ such that $(1 + \varepsilon')/(1 - \varepsilon') \leq 1 + \varepsilon$ for all $\varepsilon \in (0, 1)$ (the exact value of $\varepsilon'$ will be computed in Section~\ref{sec:complexity}).

If every produced subset sum of the previous step belongs to a distinct bin, then,
we can infer that the number of subsets of large elements is bounded by $n / \varepsilon^2$.
Moreover, we can prove the following lemma.

\begin{lemma} \label{lem:large_apx}
If the optimal ratio $r_{\opt}$ involves sets $S^*_1, S^*_2$ consisting of only large elements,
with $S^*_1 \cup S^*_2 = S^* \subseteq L$ and $a_n \in S^*$,
then $\Sigma (\overline{S_p}) / \Sigma (S_p) \leq (1 + \varepsilon) \cdot r_{\opt}$,
where $S_p$ is an $(1 - \varepsilon')$-apx solution to the \Part\ problem on input $S^*$.
\end{lemma}

\begin{proof}
Assume that $\Sigma (S^*_1) \leq \Sigma (S^*_2)$.
Note that sets $S_1^*, S_2^*$ are also the optimal solution of the \Part\ problem on input $S^*$.
By running an $(1 - \varepsilon')$ approximate \Part\ algorithm on input set $S^*$, we obtain the sets $S_1, S_2$ with $\Sigma (S_1) \leq \Sigma (S_2)$, where $S_1 = S_p$ and $S_2 = \overline{S_p}$.
Then,
\begin{align*}
    \frac{\Sigma(S_2)}{\Sigma(S_1)} &\leq
    \frac{\Sigma(S^*_2) + \varepsilon' \cdot \Sigma(S^*_1)}{(1 - \varepsilon') \Sigma(S^*_1)}\\
    &\leq \frac{\Sigma(S^*_2) + \varepsilon' \cdot \Sigma(S^*_2)}{(1 - \varepsilon') \Sigma(S^*_1)}\\
    &= \frac{1 + \varepsilon'}{1 - \varepsilon'} \cdot \frac{\Sigma(S^*_2)}{\Sigma(S^*_1)}\\
    &\leq (1 + \varepsilon) \cdot r_{\opt}
\end{align*}
where we used the fact that $(1 - \varepsilon') \cdot \Sigma(S^*_1) \leq \Sigma(S_1)$ as well as $\Sigma(S_2) \leq \Sigma(S^*_2) + \varepsilon' \cdot \Sigma(S^*_1)$.
\end{proof}

Therefore, we have proved that when the optimal solution consists of sets comprised of only large elements, it is possible to find an ($1 + \varepsilon$)-approximation solution
for the constrained \SSR\ problem by running an $(1 - \varepsilon')$-approximation algorithm for \Part\ with input the union of said large elements.
In order to do so, it suffices to consider as input all the $2^{\abs{L}-1}$ subsets of $L$ containing $a_n$ and each time run an $(1 - \varepsilon')$-approximation \Part\ algorithm.
The total cost of this procedure will be thoroughly analyzed in Section~\ref{sec:complexity} and depends on the algorithm used.

It is important to note that by utilizing an (exact or approximation) algorithm for \Part, we establish a connection between the complexities of \Part\ and approximating \SSR\ in a way that any future improvement in the first carries over to the second.

\subsection{General $(1+\varepsilon)$-Approximation Solutions} \label{subsec:plus_small}
Whereas we previously considered optimal solutions involving exclusively large elements,
here we will search for approximations for those optimal solutions that use all the elements of the input set, hence include small elements, and satisfy our constraint (i.e. $a_n$ belongs to the optimal solution sets).
We will prove that in order to approximate those optimal solutions, it suffices to consider only the $(1 - \varepsilon')$-apx solutions of the \Part\ problem corresponding to each subset of large elements and add small elements to them.
In other words, instead of considering any two random disjoint subsets consisting of large elements\footnote{Note that the number of these random pairs is $2 \cdot 3^{\abs{L}-1}$,
since $a_n$ is necessarily part of the solution.} and subsequently adding to these the small elements,
we can consider only the $(1 - \varepsilon')$-approximate solutions to the \Part\ problem computed in the previous step,
ergo, at most $B = n / \varepsilon^2$ configurations regarding the large elements.
Moreover, we will prove that it suffices to add the small elements to our solution in a greedy way.

Since the algorithm has not detected a solution so far, due to Lemma~\ref{lem:two_sets_one_bin} every computed subset sum of set $L$ belongs to a different bin.
Thus, their total number is bounded by the number of bins $B$, i.e.
\[
    2^{\abs{L}} \leq \parens*{\frac{n}{\varepsilon^2}} \iff
    \abs{L} \leq \log \parens*{\frac{n}{\varepsilon^2}}
\]
We proceed by additionally involving small elements into our solutions in order to reduce the difference between the sums of the sets, thus reducing their ratio.

\begin{lemma} \label{lem:general}
Assume that we are given the $(1 - \varepsilon')$-apx solutions for the \Part\ problem on every subset of large elements containing $a_n$.
Then, an $(1 + \varepsilon)$-approximation solution for the constrained version of \SSR\ can be found,
when the optimal solution involves small elements.
\end{lemma}

\begin{proof}
Let $S^*_1,S^*_2$ be disjoint subsets that form an optimal solution for the constrained version of \SSR, where:
\begin{itemize}
    \item $\Sigma(S^*_1) \leq \Sigma(S^*_2)$ and $a_n \in S^* = S^*_1 \cup S^*_2$.
    \item $S^*_1 = L^*_1 \cup M^*_1$ and $S^*_2 = L^*_2 \cup M^*_2$, where $L^*_1, L^*_2 \subseteq L$ and $M^*_1, M^*_2 \subseteq M$.
    \item $M^*_1 \cup M^*_2 \neq \emptyset$.
\end{itemize}
Note that, due to Lemma~\ref{lem:two_sets_one_bin}, it holds that $\Sigma(L^*_1) \neq \Sigma(L^*_2)$.
Moreover, let $L^*_p$ and $\overline{L^*_p}$ be the optimal solution of the \Part\ problem on input $L_1^* \cup L_2^*$,
while $L_p$ and $\overline{L_p}$ be the sets returned by an $(1 - \varepsilon')$-apx algorithm.

In this case, it holds that:
\begin{itemize}
    \item $\Sigma(L^*_p) \leq \Sigma(\overline{L^*_p})$ and $\Sigma(\overline{L^*_p}) - \Sigma(L^*_p) \leq \abs{\Sigma (L^* \setminus X) - \Sigma(X)},
    \forall X \subseteq L^* = L^*_1 \cup L^*_2$.
    \item $(1 - \varepsilon') \cdot \Sigma(L^*_p) \leq \Sigma(L_p) \leq \Sigma(L^*_p)$.
    \item $\Sigma(\overline{L^*_p}) \leq \Sigma(\overline{L_p}) \leq \Sigma(\overline{L^*_p}) + \varepsilon' \cdot \Sigma(L^*_p) \leq (1 + \varepsilon') \cdot \Sigma(\overline{L^*_p})$.
    \item $a_n \leq \Sigma(\overline{L^*_p})$, since $a_n$ is an element of the input set.
\end{itemize}

\noindent\textbf{Case 1.} Suppose that $\Sigma(L^*_1) > \Sigma(L^*_2)$.
Then, $\Sigma(L^*_1) - \Sigma(L^*_2) \leq \Sigma(M)$, since
\[
    \Sigma(L^*_1) \leq 
    \Sigma(S^*_1) \leq
    \Sigma(S^*_2) =
    \Sigma(L^*_2) + \Sigma(M^*_2) \leq
    \Sigma(L^*_2) + \Sigma(M).
\]

Additionally, it holds that $\Sigma(\overline{L^*_p}) \leq \Sigma(L^*_1)$ and $\Sigma(L^*_p) \geq \Sigma(L^*_2)$,
therefore $\Sigma(\overline{L^*_p}) - \Sigma(L^*_p) \leq \Sigma(M)$ follows.

\noindent\textbf{Case 1.a.} Let $\Sigma(L_p) + \Sigma(M) \geq \Sigma(\overline{L_p})$.
In this case, there exists $k$ such that $M_k = \setdef{a_i \in M}{i \in [k]} \subseteq M$ and also $0 \leq \Sigma(L_p \cup M_k) - \Sigma(\overline{L_p}) \leq \varepsilon \cdot a_n$ (since all small elements have value less than $\varepsilon \cdot a_n$).
Notice that $\Sigma(\overline{L_p}) \geq a_n$.
Hence
\[
    1 \leq \frac{\Sigma (L_p \cup M_k)}{\Sigma(\overline{L_p})} \leq
    1 + \frac{\varepsilon \cdot a_n}{\Sigma (\overline{L_p})} \leq
    1 + \frac{\varepsilon \cdot a_n}{a_n} = 1 + \varepsilon.
\]

\noindent\textbf{Case 1.b.} Let $\Sigma(L_p) + \Sigma(M) < \Sigma(\overline{L_p})$.
It suffices to show that
\[
    \frac{\Sigma (\overline{L_p})}{\Sigma(L_p \cup M)} \leq
    1 + \varepsilon.
\]
By assuming the contrary, it holds that
\begin{align*}
    \Sigma (\overline{L_p}) &>
    (1 + \varepsilon) \cdot \Sigma(L_p \cup M) \\
    &= (1 + \varepsilon) \cdot \parens*{\Sigma(L_p) + \Sigma(M)} \\
    &\geq (1 + \varepsilon)(1 - \varepsilon') \cdot \Sigma(L^*_p) + (1 + \varepsilon) \cdot \Sigma(M) \\
    &\geq (1 + \varepsilon') \cdot \Sigma(L^*_p) + (1 + \varepsilon) \cdot \Sigma(M) \\
    &\geq (1 + \varepsilon') \cdot (\Sigma(L^*_p) + \Sigma(M)).
\end{align*}
Additionally, it holds that $(1 + \varepsilon') \cdot \Sigma (\overline{L^*_p}) \geq \Sigma (\overline{L_p})$,
thus
\[
    \Sigma (\overline{L^*_p}) > \Sigma(L^*_p) + \Sigma(M)
\]
follows, which is a contradiction since $\Sigma(\overline{L^*_p}) - \Sigma(L^*_p) \leq \Sigma(M)$.

\noindent\textbf{Case 2.} Suppose that $\Sigma(L^*_1) < \Sigma(L^*_2)$.
Now there are two cases:

\noindent\textbf{Case 2.a.} Suppose that $\Sigma (L^*_1) + \Sigma(M) \leq \Sigma (L^*_2)$.
Then, it holds that $S_1^* = L_1^* \cup M$ and $S_2^* = L_2^*$ and one can infer that
$L_1^* = L^*_p$ and $L_2^* = \overline{L^*_p}$. Then,
\begin{align*}
    \frac{\Sigma(\overline{L_p})}{\Sigma(L_p \cup M)} &=
    \frac{\Sigma(\overline{L_p})}{\Sigma(L_p) + \Sigma(M)}\\
    &\leq \frac{(1 + \varepsilon') \cdot \Sigma (L^*_2)}{ (1 - \varepsilon') \cdot \Sigma (L^*_1) + \Sigma(M)}\\
    &\leq \frac{1 + \varepsilon'}{1 - \varepsilon'} \cdot \frac{\Sigma (L^*_2)}{\Sigma (L^*_1) + \Sigma(M)}\\
    &\leq (1 + \varepsilon) \cdot \frac{\Sigma (S^*_{2})}{\Sigma (S^*_{1})}
\end{align*}

\noindent\textbf{Case 2.b.} Suppose that $\Sigma (L^*_1) + \Sigma(M) > \Sigma (L^*_2)$.
We will proceed in a similar way as in cases 1.a. and 1.b.

\noindent\textbf{Case 2.b.i.} Let $\Sigma (L_p) + \Sigma(M) \geq \Sigma (\overline{L_p})$.
In this case, there exists $k$ such that $M_k = \setdef{a_i \in M}{i \in [k]} \subseteq M$ and also $0 \leq \Sigma(L_p \cup M_k) - \Sigma(\overline{L_p}) \leq \varepsilon \cdot a_n$ (since all small elements have value less than $\varepsilon \cdot a_n$).
Notice that $\Sigma(\overline{L_p}) \geq a_n$.
Hence
\[
    1 \leq \frac{\Sigma (L_p \cup M_k)}{\Sigma(\overline{L_p})} \leq
    1 + \frac{\varepsilon \cdot a_n}{\Sigma (\overline{L_p})} \leq
    1 + \frac{\varepsilon \cdot a_n}{a_n} = 1 + \varepsilon.
\]

\noindent\textbf{Case 2.b.ii.} Let $\Sigma(L_p) + \Sigma(M) < \Sigma(\overline{L_p})$.
It suffices to show that
\[
    \frac{\Sigma (\overline{L_p})}{\Sigma(L_p \cup M)} \leq
    1 + \varepsilon.
\]
By assuming the contrary, it holds that
\begin{align*}
    \Sigma (\overline{L_p}) &>
    (1 + \varepsilon) \cdot \Sigma(L_p \cup M) \\
    &= (1 + \varepsilon) \cdot \parens*{\Sigma(L_p) + \Sigma(M)} \\
    &\geq (1 + \varepsilon)(1 - \varepsilon') \cdot \Sigma(L^*_p) + (1 + \varepsilon) \cdot \Sigma(M) \\
    &\geq (1 + \varepsilon') \cdot \Sigma(L^*_p) + (1 + \varepsilon) \cdot \Sigma(M) \\
    &\geq (1 + \varepsilon') \cdot (\Sigma(L^*_p) + \Sigma(M)).
\end{align*}
Additionally, it holds that $(1 + \varepsilon') \cdot \Sigma (\overline{L^*_p}) \geq \Sigma (\overline{L_p})$,
thus
\[
    \Sigma (\overline{L^*_p}) > \Sigma(L^*_p) + \Sigma(M)
\]
follows, which is a contradiction since $\Sigma(\overline{L^*_p}) - \Sigma(L^*_p) < \Sigma(M)$,
due to the fact that $\Sigma(L^*_2) - \Sigma(L^*_1) < \Sigma(M)$ and $\Sigma(\overline{L^*_p}) - \Sigma(L^*_p) \leq \Sigma(L^*_2) - \Sigma(L^*_1)$.
\end{proof}

\subsubsection*{Adding Small Elements Efficiently.}
Here, we will describe a method to efficiently add small elements to our sets.
In particular, we search for some $k$ such that $0 \leq \Sigma(L_p \cup M_k) - \Sigma(\overline{L_p}) \leq \varepsilon \cdot a_n$, where $M_k = \setdef{a_i \in M}{i \in [k]}$.
Notice that if $\Sigma (M) \geq \Sigma(\overline{L_p}) - \Sigma(L_p)$, there always exists such a set $M_k$, since by definition, each element of set $M$ is smaller than $\varepsilon \cdot a_n$.
In order to determine $M_k$, we make use of an array of partial sums $T[k] = \Sigma (M_k)$,
where $k \leq \abs{M}$.
Since $T$ is sorted, we can find $k$ in $\bO (\log \abs{L}) = \bO (\log n)$ using binary search.



\section{Final Algorithm} \label{sec:algo}
The algorithm presented in the previous section constitutes an approximation scheme for \SSR\ when one of the solution subsets contains the maximum element of the input set.
Thus, in order to solve the \SSR\ problem, it suffices to run the previous algorithm $n$ times, where $n$ depicts the cardinality of the input set $A$, while each time removing the max element of $A$.

In particular, suppose that the optimal solution involves disjoint sets $S_1^*$ and $S_2^*$, where $a_k = \max (S_1^* \cup S_2^*)$.
There exists an iteration for which the algorithm considers as input the set $A_k = \braces{a_1, \ldots, a_k}$.
In this iteration, the element $a_k$ is the largest element and the algorithm searches for an approximation of the optimal solution for which $a_k$ is contained in one of the solution subsets.
The optimal solution of the unconstrained version of \SSR\ has this property so the ratio of the approximate solution that the algorithm of the previous section returns is at most $(1 + \varepsilon)$ times the optimal.

Consequently, $n$ repetitions of the algorithm suffice to construct an FPTAS for \SSR.

Notice that if at some repetition,
the sets returned due to the algorithm of Section~\ref{sec:algo_restricted} have ratio at most $1 + \varepsilon$,
then this ratio successfully approximates the optimal ratio $r_{\opt} \geq 1$, since $1 + \varepsilon \leq (1 + \varepsilon) \cdot r_{\opt}$,
therefore they constitute an approximation solution.

\begin{algorithm}[htb]
    \caption{SSR($A, \varepsilon$)}
    \label{alg:final}
    \begin{algorithmic}[1]
        \Require Sorted set $A = \braces{a_1, \ldots, a_n}$ and error margin $\varepsilon$.
    	\Ensure $(1 + \varepsilon)$-apx of the optimal solution for \SSR.

        \State Create array $T$ such that $T[k] = \sum_{i=1}^k a_i$. \Comment{$\Theta(n)$ time.}
        \For{$i = n, \ldots, 1$}
            \State ConstrainedSSR($\braces{a_1, \ldots, a_i}, \varepsilon, T$)
        \EndFor
    \end{algorithmic}
\end{algorithm}



\section{Complexity}
\label{sec:complexity}

The total complexity of the final algorithm is determined by three distinct operations, over the $n$ iterations of the algorithm:
\begin{enumerate}
    \item The cost to compute all the possible subset sums occurring from large elements.
    It suffices to consider the case where this is bounded by the number of bins $B = n / \varepsilon^2$, due to Lemma~\ref{lem:two_sets_one_bin}.
    
    \item The cost to compute an $(1 - \varepsilon')$-apx solution for \Part\ on each subset of large elements.
    The cost of this operation will be analyzed in the following subsection.

    \item The cost to include small elements to the $(1 - \varepsilon')$-apx solutions for \Part.
    There are $B$ such solutions, and each requires $\bO (\log n)$ time, thus the total time required is $\bO \parens*{\frac{n}{\varepsilon^2} \cdot \log n}$.
\end{enumerate}

\subsection{Complexity of Partition Computations}

\subsubsection*{Using Exact Partition Computations.}
Firstly, we will consider the case where we compute the optimal solution of the \Part\ problem.
In order to do so, we will use the standard meet in the middle algorithm~\cite{HorowitzS74} for \SubS,
and in the following we analyze its complexity.

Let subset $L' \subseteq L$ such that $\abs{L'} = k$.
The meet in the middle algorithm on the set $L'$ costs time
\[
    \bO \parens*{2^{\abs{L'} / 2} \cdot \abs{L'}}
\]
Notice that the number of subsets of $L$ of cardinality $k$ is $\binom{\abs{L}}{k}$ and that $\abs{L} \leq \log (n / \varepsilon^2)$.
Furthermore, let $c = \parens*{1 + \sqrt{2}}$, where $\log c = 1.271... < 1.3$.
Then, it holds that
\begin{align*}
    \sum_{k = 0}^{\abs{L}} \binom{\abs{L}}{k} \cdot 2^{k / 2} \cdot k
    &\leq \abs{L} \cdot \sum_{k = 0}^{\abs{L}} \binom{\abs{L}}{k} \cdot 2^{k / 2}\\
    &= \abs{L} \cdot \parens*{1 + \sqrt{2}}^{\abs{L}}\\
    &= \abs{L} \cdot c^{\abs{L}}\\
    &\leq \log (n / \varepsilon^2) \cdot c^{\log (n / \varepsilon^2)}\\
    &= \log (n / \varepsilon^2) \cdot (n / \varepsilon^2)^{\log c}
\end{align*}
where we used the Binomial Theorem.
Consequently, the complexity to solve the \Part\ problem for all the subsets of large elements is
\[
    \bO \parens*{\frac{n^{1.3}}{\varepsilon^{2.6}} \cdot \log (n / \varepsilon^2)}
\]

\subsubsection*{Using Approximate Partition Computations.}
Here we will analyze the complexity in the case we run an approximate \Part\ algorithm
in order to compute the $(1 - \varepsilon')$-approximation solutions.

For subset $L' \subseteq L$, we run an approximate \Part\ algorithm with error margin $\varepsilon'$ such that
\[
    \frac{1 + \varepsilon'}{1 - \varepsilon'} \leq 1 + \varepsilon \iff \varepsilon' \leq \frac{\varepsilon}{2 + \varepsilon}
\]
and by choosing the maximum such $\varepsilon'$, it holds that
\[
    \varepsilon' = \frac{\varepsilon}{2 + \varepsilon} \implies
    \frac{1}{\varepsilon'} = \frac{2 + \varepsilon}{\varepsilon} =
    \frac{2}{\varepsilon} + 1 \implies
    \frac{1}{\varepsilon'} = \bO \parens*{\frac{1}{\varepsilon}}
\]
Since there are at most $n / \varepsilon^2$ subsets of large elements,
we will need to run said algorithm at most $n / \varepsilon^2$ times on $\abs{L'} \leq \abs{L}$ elements and with error margin $\varepsilon'$.

Note that any approximate \SubS\ algorithm could be used in order to approximate \Part,
such as the one presented by Kellerer \etal~\cite{KellererMPS03} of complexity $\bO \parens*{\min \braces*{\frac{n}{\varepsilon}, n + \frac{1}{\varepsilon^2} \cdot \log(1 / \varepsilon)}}$.
In our case, with $\abs{L} = \log (n / \varepsilon^2)$ and error margin $\varepsilon'$, the total complexity is
\begin{align*}
    \bO \parens*{\frac{n}{\varepsilon^2} \cdot \min \braces*{\frac{\abs{L}}{\varepsilon'}, \abs{L} + \frac{1}{(\varepsilon')^2} \cdot \log (1 / \varepsilon')}} =\\
    \bO \parens*{\frac{n}{\varepsilon^2} \cdot \min \braces*{\frac{\log (n / \varepsilon^2)}{\varepsilon}, \log (n / \varepsilon^2) + \frac{1}{\varepsilon^2} \cdot \log (1 / \varepsilon)}}
\end{align*}

Very recently, Bringmann and Nakos~\cite[Theorem 5.1]{BringmannN21} developed a dedicated algorithm for approximating \Part\ via min-plus convolution computations.
By using the classic min-plus convolution algorithm of complexity $\bO (n^2)$, one can therefore obtain a deterministic FPTAS for \Part\ running in time
\[
    \bO \parens*{\abs{L} + (1 / \varepsilon')^{3/2} \cdot \log^2 \Big( \frac{\abs{L}}{\varepsilon'} \Big)}
\]
Consequently, by using this algorithm, the final complexity due to all the approximate \Part\ computations is
\begin{align*}
    \bO \parens*{\frac{n}{\varepsilon^2} \cdot \parens*{\log (n / \varepsilon^2) + (1 / \varepsilon)^{3/2} \cdot \log^2 \Big( \frac{\log (n / \varepsilon^2)}{\varepsilon} \Big)}}
\end{align*}

\subsection{Total Complexity}
The total complexity of the algorithm occurs from the $n$ distinct iterations required and depends on the algorithm chosen to find the (exact or approximate) solution to the \Part\ problem,
since all of the presented algorithms dominate the time of the rest of the operations.
Thus, by choosing the fastest one (depending on the relationship between $n$ and $\varepsilon$), the final complexity is
\begin{align*}
    \bO \parens*{\min \braces*{
    \frac{n^{2.3}}{\varepsilon^{2.6}} \cdot \log ( n / \varepsilon^2 ),
    \frac{n^2}{\varepsilon^3} \cdot \log \frac{n}{\varepsilon^2},
    \frac{n^2}{\varepsilon^2} \parens*{\log \frac{n}{\varepsilon^2} + \frac{1}{\varepsilon^{1.5}} \cdot \log^2 \frac{\log (n / \varepsilon^2)}{\varepsilon}}
    }}
\end{align*}



\section{Conclusion and Future Work} \label{sec:future}

The main contribution of this paper, apart from the introduction of a new FPTAS for the \SSR\ problem, is the establishment of a connection between \Part\ and approximating \SSR.
In particular, we showed that any improvement over the classic meet in the middle algorithm~\cite{HorowitzS74}, or over the approximation scheme for \Part\ will result in an improved FPTAS for \SSR.

Additionally, we establish that the complexity of approximating \SSR, expressed in the form $\bO \parens{\parens{n + 1 / \varepsilon}^c}$ has an exponent $c < 5$, which is an improvement over all the previously presented FPTASs for the problem.

It is important to note however, that there is a distinct limit to the complexity that one may achieve for the \SSR\ problem using the techniques discussed in this paper.

As a direction for future research, we consider the use of exact \SubS\ or \Part\ algorithms parameterized by a \emph{concentration} parameter $\beta$, as described in~\cite{AustrinKKN15,AustrinKKN16}, where they solve the decision version of \SubS.
See also~\cite{DuttaR22} for a use of this parameter under a pseudopolynomial setting.
It would be interesting to investigate whether analogous arguments could be used to solve the optimization version.


%
\bibliographystyle{splncs04}
\bibliography{bibliography}
\end{document}